\documentclass[12pt]{amsart}

\usepackage[hmargin=0.8in,height=8.6in]{geometry}
\usepackage{amssymb,amsthm, times}
\usepackage{delarray,verbatim}
\usepackage{srcltx}
\usepackage{ifpdf}
\ifpdf
\usepackage[pdftex]{graphicx}
 \else
\usepackage[dvips]{graphicx}
 \fi

\usepackage{ifpdf}
\usepackage{color}
\definecolor{webgreen}{rgb}{0,.5,0}
\definecolor{webbrown}{rgb}{.8,0,0}
\definecolor{emphcolor}{rgb}{0.95,0.95,0.95}

\usepackage{hyperref}
\hypersetup{%
          colorlinks=true,
          linkcolor=webbrown,
          filecolor=webbrown,
          citecolor=webgreen,
          breaklinks=true}
\ifpdf \hypersetup{pdftex,
            pdfstartview=FitH, %%Fit, FitB, FitH
            bookmarksopen=true,
            bookmarksnumbered=true
} \else \hypersetup{dvips} \fi

\numberwithin{equation}{section}
\newtheorem{proposition}{Proposition}[section]

\newtheorem{remark}{Remark}[section]

\newtheorem{assump}{Assumption}[section]

\newcommand {\R}{\mathbb{R}}
\newcommand {\Fb}{\mathbb{F}}

\newcommand {\F}{\mathcal{F}}

\newcommand {\p}{\mathbb{P}}

\newcommand {\E}{\mathbb{E}}

\newcommand{\diff}{{\rm d}}

\newcommand{\conn}{\quad\text{and}\quad}
\newcommand{\1}{\mbox{1}\hspace{-0.25em}\mbox{l}}
\newcommand{\lev}{L\'{e}vy }

\title{An Excursion-Theoretic Approach to Regulator's Bank Reorganization Problem}
\author[M. Egami]{Masahiko Egami}
\address[M. Egami]{Graduate School of Economics,
Kyoto University, Sakyo-Ku, Kyoto, 606-8501, Japan}
\email{egami@econ.kyoto-u.ac.jp}
\urladdr{http://www.econ.kyoto-u.ac.jp/{\textasciitilde}egami/}
\thanks{First draft: \today.  \\ This work is in part supported
by Grant-in-Aid for Scientific Research (B) No.\ 22330098, Japan Society for the Promotion of Science.}
\author[T. Oryu]{Tadao Oryu}
\address[T. Oryu]{Graduate School of Economics,
Kyoto University, Sakyo-Ku, Kyoto, 606-8501, Japan}
\email{oryu.tadao.27r@st.kyoto-u.ac.jp}

\date{}
\begin{document}

\begin{abstract}
The importance of  the global financial system cannot be exaggerated.  When a large financial institution becomes problematic and is bailed out, that bank is often claimed as ``too big to fail".  On the other hand, to prevent bank's failure, regulatory authorities adopt the Prompt Corrective Action (PCA) against a bank that violates certain criteria, often measured by its leverage ratio.  In this article, we provide a framework where one can
analyze the cost and effect of PCA's.  We model a large bank with deteriorating asset and regulatory actions attempting to prevent a failure.  The model uses the excursion theory of \lev processes and finds an optimal leverage ratio that triggers a PCA. A nice feature includes it incorporates the fact that social cost associated with PCA's are be greatly affected by the size of banks subject to PCA's, so that one can see the cost of rescuing a bank ``too big to fail". 
\end{abstract}
\maketitle \noindent \small{\textbf{Key words:} Prompt corrective actions; excursion theory; spectrally negative \lev processes; scale functions
 }\\
\section{Introduction}%%%%%%%%%%%%%%%%%%%%%%%%%%%%%%%%%%%%%%%%%%%%%%%%%%%%%%%%%%%%%%%%%%%%%
For a description of the Prompt Corrective Action (PCA, thereafter) we first quote from Shibut et al \cite{pca}: The Prompt Corrective Action (PCA) provisions in Federal Deposit Insurance Corporation Improvement Act of 1991 (FDICIA) require that regulators set a \emph{threshold} for critically undercapitalized institutions, and that regulators promptly close institutions that breach the threshold unless they quickly recapitalize or merge with a healthier institution.� Many economists expected these provisions to result in dramatically reduced loss rates, or even zero loss rates, for bank failures. In short, PCA provides a set of mandatory and discretionary actions to be taken by  banking supervisors when the bank's capital ratio is declining.  In many countries, as the above says, regulatory authorities set minimum capital ratio and intervene bank operations once the bank's capital falls short of the minimum requirement.

Available studies on PCA are very scarce.  Kocherlakova and Shim \cite{kocherlakota2007} and Shim \cite{shim2011} develop dynamic contract models to analyze under what conditions regulators should liquidate a problematic bank or subsidize.  While liquidation is one alternative in PCA's, we need to analyze a broader spectrum of actions including recapitalization, cash infusion, and changes of risk-return profile of the bank's asset.  Considering the catastrophic turmoil in the financial systems we experienced in the recent global financial crisis, a comprehensive analytical framework for the regulators' interventions to a problem bank's management  is very much needed.
Our model has, among other things, the following features:
\begin{itemize}
  \item we incorporate \emph{leverage ratio} explicitly  to describe the situation where even a large bank (with high leverage) can  \emph{easily} fail and to directly deal with  leverage ratio threshold that triggers PCA;
  \item we use spectrally negative \lev processes for modeling sudden declines in the value of bank assets;
  \item we include cash infusions (at the beginning of PCA) that are often used for preventing  outright insolvency, change the bank's risk-return profile, and consider the possibilities whether the bank comes back to the normal operation or goes to liquidation;
  \item we incorporate various costs associated with PCA's and compute optimal threshold level that triggers a PCA in the sense that the total cost can be minimized; and
  \item we obtain a result, among other things, where the bank size has a crucial impact on the cost involved, which well captures the real-life experience.
\end{itemize}

In this paper, we describe deterioration of leverage ratio as the excursion from the running maximum.  It should be best to explain through an example.  We shall define everything rigourously in the next section.
Let $Y=e^X$ be the bank's total asset value, where $X$ is a spectrally negative \lev process and represents the fluctuation rate of the asset.  Let $S$ be the running maximum of $X$.  We assume  that the bank increases its asset base as long as it maintains the predetermined \emph{leverage ratio}, defined as
\[
\text{Leverage ratio}:=\frac{\text{Debt}}{\text{Total Asset}}.
\] Let us set the said leverage ratio as $e^{-b}$.
For example, if the bank has the initial asset of $Y_0=e^{X_0}=100$ with $e^{-b}=0.8$, it has total asset of $100$ financed by
debt $80$ and equity $20$.  We can think of this ratio as the maximum leverage ratio that is allowed by the banking regulations. We assume that the bank increases
its asset base as long as $X=S$ where $S$ is the running maximum of $X$ and that the bank's leverage ratio is maintained at $0.8$.  Hence if the asset value
appreciates to $120$, then this would provide the bank with more lending opportunity since the equity value is now $40$.  With this new equity level, the bank
increases its leverage up to $0.8$, that is, total asset increasing to $200$ financed by debt $160$ and equity $40$. Note that  $e^S=e^X=200$ and the debt level
is $e^{-b}(e^S)=e^{S-b}=160$. Now if the bank's asset deteriorates due to defaults in the lending portfolio, we would have $S-X>0$. In other words, there appears
an \emph{excursion from the running maximum of $S$}.  Since the asset level has been pegged at $e^S=200$, the bank's equity would be wiped out when $e^{S-b}=e^X$. That is, the process is absorbed at $t=T_b$, i.e, the first time $X$ goes below level $S-b$.
Moreover, note that this model can incorporate the regulatory requirements that the bank, when experiencing asset deterioration,  need to sell the assets in order to reduce the leverage. For example, assume that when the bank loses one dollar of asset, the bank  loses its equity by $\alpha$ and reduces its debt by $1-\alpha$, where $\alpha\in(0,1]$. Then, at the time the equity is wiped out, we have
\[
e^X\le e^S\left(1-\frac{1-e^{-b}}{\alpha}\right),
\]
that is, the process is absorbed when the excursion $S-X$ reaches to $\log\left(\frac{1-e^{-b}}{\alpha}-1\right)$.

An excursion theory for spectrally negative \lev processes has been developed recently.  See Bertoin \cite{Bertoin_1996} as a general reference.  More specifically, an exit problem of the reflected process $Y$ was studied by Avram et al. \cite{avram-et-al-2004}, Pistrorius \cite{Pistorius_2004} \cite{Pistorius_2007} and Doney \cite{Doney_2005}. For \emph{spectrally negative \lev processes}, or \lev processes with only negative jumps, a number of authors have succeeded in solving interesting stochastic optimization problems and in extending the classical results by using the \emph{scale functions}, which we shall review briefly later.  We just name a few here : \cite{Baurdoux2008,Baurdoux2009}  for stochastic games,  \cite{Avram_et_al_2007, Kyprianou_Palmowski_2007, Loeffen_2008}  for the optimal dividend problem, \cite{alili-kyp, avram-et-al-2004} for American and Russian options,  and \cite{Egami-Yamazaki-2010-1, Kyprianou_Surya_2007} for credit risk.

The rest of the paper is organized as follows.  In Section
\ref{sec:model}, we formulate a mathematical model to express PCA program and then
find an optimal threshold triggering level in Section \ref{sec:solution}.  We shall illustrate the solution through a numerical example in Section \ref{sec:example}.
Furthermore, we shall consider the situation where,  after the bank successfully emerges from the intervention (i.e., PCA ends), it again becomes problematic and subject to another PCA.  This is in Section \ref{sec:multiple}.
 We attach a brief summary of results from the theory of \emph{scale functions} associated with spectrally negative \lev processes in Appendix.

\section{Mathematical Model}\label{sec:model} %%%%%%%%%%%%%%%%%%%%%%%%%%%%%%%%%%%%%%%%%%%%%

Let the spectrally negative Levy processes $X^i=\{X^i_t;t\geq 0\}$ $(i=0,1)$ represent the state variable defined on the probability space
$(\Omega, \F, \p)$, where $\Omega$ is the set of all possible realization of the
stochastic economy, and $\p$ is a probability measure defined on $\F$. We denote by
$\mathbb{F}=\{\F_t\}_{t\ge 0}$ the filtration with respect to which $X^0$ and $X^1$ are adapted and with the usual
conditions being satisfied. The Laplace exponent $\psi_i$ of $X^i$, i=0,1, is given by
\[
\psi_i(\lambda)=\mu_i\lambda+\frac{1}{2}\sigma_i^2\lambda^2+\int_{(-\infty,0)}(e^{\lambda x}-1-\lambda x \1_{(x>-1)})\Pi_i(\diff x),
\]
where $\mu_i \geq 0$, $\sigma_i \geq 0$, and $\Pi_i$ is a measure concentrated on $\R\backslash \{0\}$ satisfying
$\int_{\R}(1\wedge x^2)\Pi_i(\diff x)<\infty$. It is well-known that $\psi_i$ is zero at the origin and convex on $\R_+$.

We define the process $X=\{X_t;t\geq 0\}$ as the solution
to the stochastic differential equation
\[
\diff X_t=\diff X^{I(t)}_t \conn X_0=x,
\]
where $I=\{I(t);t\geq 0\}$ is the right-continuous switching process which satisfies
$I(t)\in\{0,1\}$ for every $t\in\R_+$. We postpone (see \eqref{eq:I}) the rigorous mathematical definition of the process $I$ to make the explanation of our model smoother.

The bank's total asset value is represented by the process $Y=\{e^{X_t};t\geq 0\}$. Therefore, $X$ represents the fluctuation rate of the bank's total asset value. When $I(t)=0$,
the bank is well capitalized with satisfactory leverage ratio and thus is not subject to the regulator's PCA. Our $\diff X^0$ corresponds to
the dynamics while not being controlled. On the other hand, when $I(t)=1$, a PCA is applied and the bank is taken into
strict supervision by the regulator with  corresponding asset fluctuation rate $\diff X^1$.  In general, it may be often the case that
\[
\mu_1<\mu_0 \conn \sigma_1<\sigma_0
\]
to reflect more conservative risk-return profile during the PCA period.
We introduce $\Fb$-stopping times $\tau^+_c$ and $\tau^-_c$ $(c\geq 0)$ defined by
\[
\tau^+_{c}=\inf \{t\geq 0 : X^1_t \geq c\},\conn\tau^-_{c}=\inf \{t\geq 0 : X^1_t \leq c\}.
\]
In addition, let $S=\{S_t;t\geq 0\}$ be defined by $S_t=\sup _{u \in [0,t]} X_u\vee s$, and we introduce the $\Fb$-stopping times $T_c$ ($c>0$) defined by
\[
T_c=\inf \{t\geq 0 : S_t-X_t \geq c\}.
\]

We assume that PCA is applied (i.e. the process $I$ changes form $0$ to $1$) at $t=T_{b'}$, where $b\geq b'\geq 0$. Note that this is the time when the bank's leverage ratio $e^{S-b}/e^X$ exceeds the level $e^{b'-b}$ (not the level $e^{-b'}$). Indeed, since  $X<S-b'$, the leverage ratio is
\[
\frac{e^{S-b}}{e^X} > \frac{e^{S-b}}{e^{S-b'}}=e^{b'-b}.
\]
This threshold $b'$ is determined by the regulator and we call it the \emph{PCA trigger level}.  One of our main problems is to  calculate of the total cost involved in PCA's until the bank's lifetime.
When the bank undergoes a PCA, one of the two scenarios is possible: the bank becomes insolvent ($S-X\geq b$), or the bank successfully improves its leverage ratio to $e^{-b}$ ($S-X=0$).

Additionally, we assume that when a PCA is applied, the bank's asset is pushed up to the target level
\[e^{S_{T_{b'}}-a}, \quad \text{$a\in[0,b)$ is some constant}.\]
This is done by injecting funds (taxpayers' money) to improve the leverage ratio to some predetermined level $e^{a-b}$.  More specifically, since the PCA moves the level of $X$ to $S_{T_{b'}}-a$ and the leverage ratio at the time of PCA is $S_{T_{b'}-b}$, then the new leverage ratio becomes
\[
\frac{e^{S_{T_{b'}-b}}}{e^{S_{T_{b'}}-a}}=e^{a-b}.
\]
The amount of money to be poured (as equity) is thereby
the difference between the asset values before and after PCA is applied.

To compute the initial cost to be paid, we need to record both of the asset values before and after being pushed up. Therefore, we define the random variable $\underline{X}$ as $\underline{X}=X_{T_{b'}}$ and, afterward, redefine $X_{T_{b'}}$ by $X_{T_{b'}}=S_{T_{b'}}-a$. That is, $e^{\underline{X}}$ represents the asset value \emph{before} the cash infusion is made, and $e^{X_{T_{b'}}}$ indicates the asset value \emph{after} the bank receives fresh money. This way, $e^X$ remains representing the asset value of the bank, and we can represent the initial cash to be poured by
\[e^{S_{T_{b'}}-a}-e^{\underline{X}}.\]  Let us emphasize that this value is large when the size of the bank is large (with a large value of $e^S$).
Note that though we use this form of initial cost throughout this paper, the initial cost can be generalized to the form of  $c(\underline{X},S_{T_{b'}})$, where $c:\R^2_+\mapsto\R_+$, and in those cases, this cost may include, for example, the cash to be set aside for the worst scenario of the bank's insolvency (the bank's depositors will be bailed out by the FDIC),
or the present value of administration costs to alter bank's risk-return profile from $\psi_0$ to $\psi_1$.

Let an $\mathbb{F}$-stopping time $\tau$ be the time PCA ends, then $\tau$ should be represented by $\tau=T_{b'}+(\tau^+_a\wedge \tau^-_{a-b})\circ \theta_{T_{b'}}$, where $(\theta_t)_{t\in\R_+}$ is shift operator, and the process $I$ can be defined as
\begin{align}\label{eq:I}
I(t) =
\begin{cases}
0\quad \text{for}\quad t<T_{b'}, \tau \leq t\\
1\quad \text{for}\quad t\in[T_{b'},\tau).
\end{cases}
\end{align}
In summary, at time $t=T_{b'}$, the bank's leverage ratio becomes worse than the PCA trigger level, then a PCA starts and the bank goes under the regulator's control. The corresponding excursion height is $S_{T_{b'}}-\underline{X}$. Then the regulatory authority injects in the amount of $e^{S_{T_{b'}}-a}-e^{\underline{X}}$ and the bank's leverage ratio is improved to $e^{a-b}$.
To recover its leverage ratio to $e^{-b}$, $X^1$ must go up in the amount of $a$, and the time is denoted by $\tau^+_a$. However, if $X^1$ goes down
in the amount of $b-a$, the bank becomes insolvent, and the corresponding time is $\tau^-_{a-b}$.

In addition to the initial cost, there will be the running cost over time while PCA continues and the social cost (the penalty) when the bank becomes insolvent. We assume that (1) the running cost will be incurred in proportion to the duration of PCA being in place and the ratio is $\alpha\geq0$, and (2) the social cost to the economy caused by the bank's final insolvency is determined by $\beta e^{S_{T_b}-b}$ ($\beta\geq0$), that is,  some parameter $\beta$ times the bank's asset when it becomes insolvent. This represents the fact that the cost of bank's failure gets larger as the size of the bank is large.  Note that this penalty function may be generalized as the initial cost function.
Finally,  the expected cost $C_1$ for the PCA can be represented by
\begin{equation}\label{eq:total-cost}
C_1(x,s;b')=\E^{x,s}\left[e^{-qT_{b'}}\left(e^{S_{T_{b'}}-a}-e^{\underline{X}}\right)+\alpha\int_{T_{b'}}^{\tau}e^{-qt}\diff t + e^{-q\tau} \left(\beta e^{S_{T_{b'}}-b}\right)\1_{\{\tau^+_{a}\circ \theta_{T_{b'}} > \tau^-_{a-b}\circ \theta_{T_{b'}}\}}\right],
\end{equation} which we shall calculate in the next section.  Moreover, we analyze the cost minimizing PCA trigger level given the values of $a$.

\section{Solution}\label{sec:solution}%%%%%%%%%%%%%%%%%%%%%%%%%%%%%%%%%%%%%%%%%%%%%%%%%%%%%%%%%%%%%%%%%%%%%%%%%%%%%%%%%%%%%%%%%%%
To solve the problem, we divide the cost function $C_1$ into three blocks;
\[
C_1(x,s;b')=
\begin{cases}
C^0_1(x,s;b') & \text{ if } s-x>b',\\
C^1_1(x,s;b') & \text{ if } s-x=0,\\
C^2_1(x,s;b') & \text{ if } s-x\in(0,b').
\end{cases}
\]
As the first step, we calculate $C^0_1(x,s;b')$, the cost involved in the PCA when $S_0-X_0\geq b'$; in other words, PCA is applied at time $t=0$. In particular, this is always the case when $b'=0$. The explanation of scale functions in the following lemma will be attached in the appendix.
Note that the scale functions of $X^0$ and $X^1$ are explicitly known in some cases including the case they have no jumps
(see Hubalek and Kyprianou \cite{Hubalek_Kyprianou_2009} for example).

\begin{proposition}\label{C1_0}
If $S_0-X_0\geq b'$, then
\begin{eqnarray}\label{eq:C1_0}
C^0_1(x,s;b')&=&e^{s-a}-e^x+\frac{1}{q}\left(1
-Z_1^{(q)}(b-a)-(1-Z_1^{(q)}(b))\frac{W_1^{(q)}(b-a)}{W_1^{(q)}(b)}\right)\\&&+\beta e^{s-b} \left(Z_1^{(q)}(b-a)-Z_1^{(q)}(b)\frac{W_1^{(q)}(b-a)}{W_1^{(q)}(b)}\right),\nonumber
\end{eqnarray}
where $W^{(q)}_i$, $i=1,2$, is $q$-scale function of $X^i$, and
\[
Z_i^{(q)}(x)=1+q\int^x_0W_i^{(q)}(y)\diff y.
\]
\end{proposition}
\begin{proof}
Since $S_0-X_0\geq b'$, we have $T_{b'}=0$ and
\begin{eqnarray*}
\E^{x,s}\left[\int_{T_{b'}}^{\tau}e^{-qt}\diff t \right]
&=&\E^{s-a,s}\left[\int_0^{\tau^+_{a}\wedge \tau^-_{a-b}}e^{-qt}\diff t \right]\\
&=&\frac{1}{q}\left(1-\E^{s-a,s}\left[e^{-q(\tau^+_a\wedge \tau^-_{a-b})}\right] \right)\\
&=&\frac{1}{q}\left(1-\E^{s-a,s}\left[\1_{\{\tau^+_{a}<\tau^-_{a-b}\}}e^{-q\tau^+_{a}}\right]
-\E^{s-a,s}\left[\1_{\{\tau^+_{a}>\tau^-_{a-b}\}}e^{-q\tau^-_{a-b}}\right]\right).
\end{eqnarray*}
In the same way, we have
\begin{eqnarray*}
\E^{x,s}\left[e^{-q\tau} \left(\beta e^{S_{T_{b'}}-b}\right)\1_{\{\tau^+_{a}\circ \theta_{T_{b'}} > \tau^-_{a-b}\circ \theta_{T_{b'}}\}}\right]
&=&\beta e^{s-b} \E^{s-a,s}\left[\1_{\{\tau^+_{a} > \tau^-_{a-b}\}}e^{-q(\tau^+_a\wedge \tau^-_{a-b})}\right]\\
&=&\beta e^{s-b} \E^{s-a,s}\left[\1_{\{\tau^+_{a} > \tau^-_{a-b}\}}e^{-q\tau^-_{a-b}}\right]
\end{eqnarray*}
Then we can write
\begin{eqnarray}\label{eq:C0-interim}
C^0_1(x,s;b')&=&e^{s-a}-e^{x}+\frac{1}{q}\left(1-\E^{s-a,s}\left[\1_{\{\tau^+_{a}<\tau^-_{a-b}\}}e^{-q\tau^+_{a}}\right]
-\E^{s-a,s}\left[\1_{\{\tau^+_{a}>\tau^-_{a-b}\}}e^{-q\tau^-_{a-b}}\right]\right)\\&&+\beta e^{s-b} \E^{s-a,s}\left[\1_{\{\tau^+_{a} > \tau^-_{a-b}\}}e^{-q\tau^-_{a-b}}\right]\nonumber.
\end{eqnarray}
It is well known (see Kyprianou \cite{Kyprianou_2006} and Doney \cite{Doney_2005}) that
\begin{eqnarray*}
\E^{s-a, s}[\1_{\{\tau^+_{a}<\tau^-_{a-b}\}}e^{-q\tau^+_{a}}]&=&\frac{W_1^{(q)}(b-a)}{W_1^{(q)}(b)}, \conn\\
\E^{s-a, s}[\1_{\{\tau^+_{a} > \tau^-_{a-b}\}}e^{-q\tau^-_{a-b}}]&=&Z_1^{(q)}(b-a)-Z_1^{(q)}(b)\frac{W_1^{(q)}(b-a)}{W_1^{(q)}(b)}.
\end{eqnarray*}
Hence we have (\ref{eq:C1_0}).
\end{proof}

Now we calculate the cost in the case that $S_0=X_0=s$ by using Proposition \ref{C1_0}.
\begin{proposition}\label{C1_1}
If $b'>0$ and $S_0=X_0=s$, then
\begin{align}\label{eq:C1_1}
C^1_1(s,s;b')&=\frac{\sigma^2}{2}\left(\frac{(W_0^{(q)'}(b'))^2}{W_0^{(q)}(b')}-W_0^{(q)''}(b')\right)\int_s^{\infty}\diff m \exp\left(-(m-s)\frac{W_0^{(q)'}(b')}{W_0^{(q)}(b')}\right)C^0_1(m-b',m;b')\\
&+\iint_E\Pi(\diff h)\diff y \left(W_0^{(q)'}(y)-\frac{W_0^{(q)'}(b')}{W_0^{(q)}(b')}W_0^{(q)}(y)\right)\nonumber\\&\times\left(\int^{\infty}_{s}\diff m\exp\left(-(m-s)\frac{W_0^{(q)'}(b')}{W_0^{(q)}(b')}\right)C^0_1(m-y+h,m;b')\right),\nonumber
\end{align}
where $E=\{ (y,h)\in\R^2:0\leq y<b', y-b<h<y-b' \}$.
\end{proposition}
\begin{proof}
As for the initial cost part of \eqref{eq:total-cost}, we have, by splitting into the case where PCA trigger level $b'$ is continuously crossed and the case where it is overshot by a downward jump,
\begin{eqnarray*}
&&\E^{s,s}\left[e^{-qT_{b'}}\left(e^{S_{T_{b'}}-a}-e^{S_{T_{b'}}-\underline{X}}\right)\right]\\
&=&\int_s^{\infty}\diff m\E^{s,s}\left[e^{-qT_{b'}}\1_{\{ S_{T_{b'}}-\underline{X}=b', S_{T_{b'}}\in \diff m \}}\right]\left(e^{m-a}-e^{m-b'}\right)\\
&&+\iiint_D\E^{s,s}\left[e^{-qT_{b'}}\1_{\{ S_{T_{b'}-}-X_{T_{b'}-}\in\diff y, \underline{X}-X_{T_{b'}-}\in \diff h, S_{T_{b'}-}\in \diff m \}}\right]\left(e^{m-a}-e^{m-y+h}\right),
\end{eqnarray*}
where $D=\{ (m,y,h)\in\R^3 : 0\leq y<b', y-b<h<y-b', m>s \}$.
Note that $X_{T_{b'}-}$ is, as usual, the pre-jump position of $X$ at time $T_{b'}$.  Because of the Markov property of $X_0$ and $X_1$, we have
\begin{eqnarray*}
\E^{s,s}\left[\int_{T_{b'}}^{\tau}e^{-qt}\diff t \right]
&=&\E^{s,s}\left[e^{-qT_{b'}}\int_0^{(\tau^+_a\wedge \tau^-_{a-b})\circ \theta_{T_{b'}}}e^{-qt}\diff t \right] \\
&=&\E^{s,s}\left[e^{-qT_{b'}}\right]\E^{s, s}\left[\int_0^{(\tau^+_a\wedge \tau^-_{a-b})\circ \theta_{T_{b'}}}e^{-qt}\diff t \bigg|\mathcal{F}_{T_{b'}}\right]\\
&=&\E^{s,s}\left[e^{-qT_{b'}}\right]\E^{s-a, s}\left[\int_0^{\tau^+_a\wedge \tau^-_{a-b}}e^{-qt}\diff t \right]\\
&=&\E^{s,s}\left[e^{-qT_{b'}}\right]\frac{1}{q}\left(1-\E^{s-a, s}\left[e^{-q(\tau^+_a\wedge \tau^-_{a-b})}\right] \right)\\
&=&\E^{s,s}\left[e^{-qT_{b'}}\right]\frac{1}{q}\left(1-\E^{s-a, s}\left[\1_{\{\tau^+_{a}<\tau^-_{a-b}\}}e^{-q\tau^+_{a}}\right]
-\E^{s-a, s}\left[\1_{\{\tau^+_{a}>\tau^-_{a-b}\}}e^{-q\tau^-_{a-b}}\right]\right)\\
&=&\bigg(\int_s^{\infty}\E^{s,s}\left[e^{-qT_{b'}}\1_{\{ S_{T_{b'}}-\underline{X}=b', S_{T_{b'}}\in \diff m \}}\right]\\
&&+\iiint_D\E^{s,s}\left[e^{-qT_{b'}}\1_{\{ S_{T_{b'}-}-X_{T_{b'}-}\in\diff y, \underline{X}-X_{T_{b'}-}\in \diff h, S_{T_{b'}-}\in \diff m \}}\right]\bigg)\\
&&\times\frac{1}{q}\left(1-\E^{m-a, m}\left[\1_{\{\tau^+_{a}<\tau^-_{a-b}\}}e^{-q\tau^+_{a}}\right]
-\E^{m-a, m}\left[\1_{\{\tau^+_{a}>\tau^-_{a-b}\}}e^{-q\tau^-_{a-b}}\right]\right).
\end{eqnarray*}
On the penalty part, we have
\begin{eqnarray*}
&&\E^{s,s}\left[e^{-q\tau} \left(\beta e^{S_{T_{b'}}-b}\right)\1_{\{\tau^+_{a}\circ \theta_{T_{b'}} > \tau^-_{a-b}\circ \theta_{T_{b'}}\}}\right]\\
&=&\E^{s,s}\left[e^{-q(T_{b'}+\tau^-_{a-b}\circ \theta_{T_{b'}}} \left(\beta e^{S_{T_{b'}}-b}\right)\1_{\{\tau^+_{a}\circ \theta_{T_{b'}} > \tau^-_{a-b}\circ \theta_{T_{b'}}\}}\right]\\
&=&\int_s^{\infty}\E^{s,s}\left[e^{-qT_{b'}}\1_{\{ S_{T_{b'}}-\underline{X}=b', S_{T_{b'}}\in \diff m \}}\right]\E^{m-a, m}\left[\beta e^{m-b} \1_{\{\tau^+_{a} > \tau^-_{a-b}\}}e^{-q\tau^-_{a-b}}\right]\\
&&+\iiint_D\E^{s,s}\left[e^{-qT_{b'}}\1_{\{ S_{T_{b'}-}-X_{T_{b'}-}\in\diff y, \underline{X}-X_{T_{b'}-}\in \diff h, S_{T_{b'}-}\in \diff m \}}\right]\E^{m-a, m}\left[\beta e^{m-b} \1_{\{\tau^+_{a} > \tau^-_{a-b}\}}e^{-q\tau^-_{a-b}}\right].
\end{eqnarray*}
Then, by summing  those three parts, we can write in view of \eqref{eq:C0-interim}
\begin{eqnarray}\label{eq:C1-in-terms-of-C0}
\hspace{1cm}C^1_1(s,s;b')&=&\int_s^{\infty}\E^{s,s}\left[e^{-qT_{b'}}\1_{\{ S_{T_{b'}}-\underline{X}=b', S_{T_{b'}}\in \diff m \}}\right]C^0_1(m-b',m;b')\\
&&+\iiint_D\E^{s,s}\left[e^{-qT_{b'}}\1_{\{ S_{T_{b'}-}-X_{T_{b'}-}\in\diff y, \underline{X}-X_{T_{b'}-}\in \diff h, S_{T_{b'}-}\in \diff m \}}\right]C^0_1(m-y+h,m;b').\nonumber
\end{eqnarray}
Finally, it is known from Theorem 1 and 2 in Pistrorius \cite{Pistorius_2005} that
\[
\E^{s,s}\left[e^{-qT_{b'}}\1_{\{ S_{T_{b'}}-\underline{X}=b', S_{T_{b'}}\in \diff m \}}\right]=\frac{\sigma^2}{2}\left(\frac{W_0^{(q)'}(b')}{W_0^{(q)}(b')}-W_0^{(q)''}(b')\right)\exp\left(-(m-s)\frac{W_0^{(q)'}(b')}{W_0^{(q)}(b')}\right)\diff m,
\]
and
\begin{eqnarray*}
&&\E^{s,s}\left[e^{-qT_{b'}}\1_{\{ S_{T_{b'}-}-X_{T_{b'}-}\in\diff y, \underline{X}-X_{T_{b'}-}\in \diff h, S_{T_{b'}-}\in \diff m \}}\right]\\
&=&\Pi(\diff h)\diff y\diff m \left(W_0^{(q)'}(y)-\frac{(W_0^{(q)'}(b'))^2}{W_0^{(q)}(b')}W_0^{(q)}(y)\right)\exp\left(-(m-s)\frac{W_0^{(q)'}(b')}{W_0^{(q)}(b')}\right).
\end{eqnarray*}
Hence we have (\ref{eq:C1_1}).
\end{proof}
Note that some condition is needed for the finiteness of $C^1_1(s,s;b')$, and the following Remark shows it.
\begin{remark}\label{convergence}
{\rm
$C^1_1(s,s;b')<\infty$ if and only if $1- \frac{W_0^{(q)'}(b')}{W_0^{(q)}(b')}<0$.
}
\end{remark}
\begin{proof}
By Proposition \ref{C1_0}, we have
\begin{align*}
C^0_1(m-u,m;b')=&\left( e^{-a}-e^{-u}+\beta e^{-b} \left(Z_1^{(q)}(b-a)-Z_1^{(q)}(b)\frac{W_1^{(q)}(b-a)}{W_1^{(q)}(b)}\right) \right)e^{m}\\&+\frac{1}{q}\left(1
-Z_1^{(q)}(b-a)-(1-Z_1^{(q)}(b))\frac{W_1^{(q)}(b-a)}{W_1^{(q)}(b)}\right),\quad \text{ for } m\geq s \text{\, and \,} u\in[b',b).
\end{align*}
Hence, there are some positive constants $M_1, M_2<\infty$ with which we can write
\begin{align*}
&\int_s^{\infty}\diff m\exp\left(-(m-s)\frac{W_0^{(q)'}(b')}{W_0^{(q)}(b')}\right)C^0_1(m-u,m;b')\\
=&\int_s^{\infty}\diff m\left(M_1\exp\left( \left( 1-\frac{W_0^{(q)'}(b')}{W_0^{(q)}(b')} \right)m \right) + M_2\exp\left(-m\frac{W_0^{(q)'}(b')}{W_0^{(q)}(b')} \right)\right),\quad \text{ for } u\in[b',b),
\end{align*}
and therefore, the integral above is finite if and only if $ 1-W_0^{(q)'}(b')/W_0^{(q)}(b')<0$.
Since $y-h\in(b',b)$ on $E$ and the other terms in \eqref{eq:C1_1} satisfy
\begin{align*}
&\frac{\sigma^2}{2}\left(\frac{(W_0^{(q)'}(b'))^2}{W_0^{(q)}(b')}-W_0^{(q)''}(b')\right)<\infty,\conn\\
\iint_E\Pi(\diff h)&\diff y \left(W_0^{(q)'}(y)-\frac{W_0^{(q)'}(b')}{W_0^{(q)}(b')}W_0^{(q)}(y)\right)<\infty\quad \text{for \,} b'>0,
\end{align*}
we can say that $C^1_1(s,s;b')<\infty$ if and only if $1-W_0^{(q)'}(b')/W_0^{(q)}(b')<0$.
\end{proof}

Finally, we calculate $C^2_1(x,s;b')$ in the case that $S_0-X_0\in(0,b')$, by using Propositions \ref{C1_0} and \ref{C1_1}.
\begin{proposition}\label{C1_2}
If $b'>0$ and $S_0-X_0\in(0,b')$, then
\begin{align}\label{eq:C1_2}
C^2_1(x,s;b')&= \frac{\sigma^2}{2}\left(W_0^{(q)}(b'-s+x)-\frac{W_0^{(q)'}(b')}{W_0^{(q)}(b')}W_0^{(q)}(b'-s+x)\right)C^0_1(s-b',s;b')\\
&+\iint_E \Pi(\diff h)\diff y\left( \frac{W_0^{(q)}(b'-s+x)}{W_0^{(q)}(b')}W_0^{(q)}(y)-W_0^{(q)}(y-s+x) \right)C^0_1(s-y+h,s;b')\nonumber\\
&+\frac{W_0^{(q)}(b-s+x)}{W_0^{(q)}(b)}C^1_1(s,s;b').\nonumber
\end{align}
\end{proposition}
\begin{proof}
In the case $S_0-X_0\in(0,b')$, two scenarios are possible. One is that $X$ reaches to $s$ before PCA is applied, and the other is that PCA applies before reaching $s$. Mathematically, this means that
\begin{eqnarray*}
C_1^2(x,s;b')&=&\E^{x,s}\left[e^{-q\tau^+_{s-x}}\1_{\{\tau^+_{s-x}<\tau^-_{s-x-b'}\}}C^1_1(s,s;b')\right]\\
&&+\E^{x,s}\left[e^{-q\tau^-_{s-x-b'}}\1_{\{\tau^+_{s-x}>\tau^-_{s-x-b'}\}}C^0_1(X_{\tau^-_{s-x-b'}},S_{\tau^-_{s-x-b'}};b')\right]\\
&=&\E^{x,s}\left[e^{-q\tau^+_{s-x}}\1_{\{\tau^+_{s-x}<\tau^-_{s-x-b'}\}}\right]C^1_1(s,s;b')\\
&&+\E^{x,s}\left[e^{-qT_{b'}}\1_{\{ S_{T_{b'}}-\underline{X}=b', S_{T_{b'}}=s \}}\right]C^0_1(s-b',s;b')\\
&&+\iint_E\E^{s,s}\left[e^{-qT_{b'}}\1_{\{ S_{T_{b'}-}-X_{T_{b'}-}\in\diff y, \underline{X}-X_{T_{b'}-}\in \diff h, S_{T_{b'}-}=s \}}\right]C^0_1(s-y+h,s;b').
\end{eqnarray*}
From Theorem 1 and 2 in Pistrorius \cite{Pistorius_2007} again, we have
\[
\E^{x,s}\left[e^{-qT_{b'}}\1_{\{ S_{T_{b'}}-\underline{X}=b', S_{T_{b'}}=s \}}\right]=\frac{\sigma^2}{2}\left(W_0^{(q)}(b'-s+x)-\frac{W_0^{(q)'}(b')}{W_0^{(q)}(b')}W_0^{(q)}(b'-s+x)\right),
\]
and
\[
\E^{s,s}\left[e^{-qT_{b'}}\1_{\{ S_{T_{b'}-}-X_{T_{b'}-}\in\diff y, \underline{X}-X_{T_{b'}-}\in \diff h, S_{T_{b'}-}=s \}}\right]=\Pi(\diff h)\diff y\left( \frac{W_0^{(q)}(b'-s+x)}{W_0^{(q)}(b')}W_0^{(q)}(y)-W_0^{(q)}(y-s+x) \right).
\]
Hence we have \eqref{eq:C1_2}.
\end{proof}
Now we have all three parts of $C_1$ and use this result for a numerical example.

\subsection{Example}\label{sec:example}
In this section, we solve a specific example. We assume $X^0$ and $X^1$ are Brownian motions with drifts and exponentially distributed jumps;
\begin{equation}\label{Example}
X_t=\mu t+\sigma B_t - \sum^{N_t}_{j=1}\epsilon_j,
\end{equation}
where $\mu, \sigma \geq 0$, $\epsilon_j$ are i.i.d. random variables which are exponentially distributed with parameter $\rho>0$ and $N_t$ is an independent Poisson
process with intensity $a>0$.
Before solving the problem, we introduce the explicit representation of the scale function for the process.
The Laplace exponent $\psi$ of $X$ has the  following simple representation;
\[
\psi(\lambda)=\frac{\sigma^2}{2}\lambda^2+\mu\lambda-\frac{a\lambda}{\rho+\lambda},\quad \lambda \geq 0.
\]
The equation $\psi(\lambda)=q$ ($q>0$) has three real solutions $\{\Phi(q),\alpha,\beta\}$ ($\Phi(q)>\alpha>\beta$), and the $q$-scale function $W^{(q)}$ of $X$ is given by
\begin{equation}\label{scalefunction}
W^{(q)}(x)=\frac{e^{\Phi(q)x}}{\psi'(\Phi(q))}+\frac{e^{\alpha x}}{\psi'(\alpha)}+\frac{e^{\beta x}}{\psi'(\beta)}.
\end{equation}

Let $\psi_i,\, i=0,1$ be defined by
\[
\psi_i(\lambda)=\frac{\sigma_i^2}{2}\lambda^2+\mu_i\lambda-\frac{a_i\lambda}{\rho_i+\lambda},\quad \lambda \geq 0.
\]
\begin{figure}[h]
\begin{center}
\begin{minipage}{0.8\textwidth}
\centering{\includegraphics[scale=0.8]{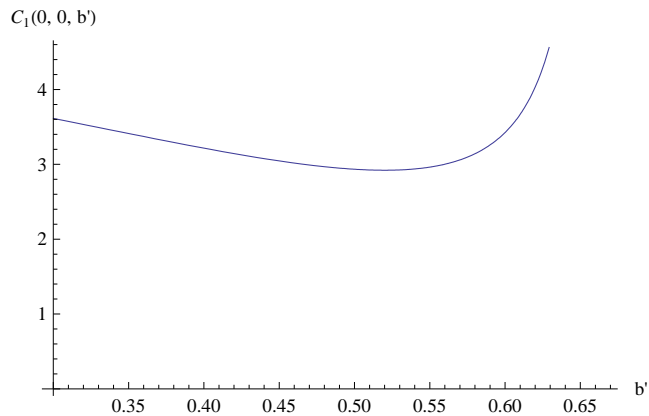}}
\end{minipage}
\caption{the graph of $C_1(0,0;b')$.}
\end{center}
\end{figure}
Figure 1 shows the graph of $C_1(0,0;b')$ with the settings $X_0=S_0=0$, $b=1$, $a=0.3$, $q=0.1$, $\alpha=\beta=1$, $\mu_0=0.2$, $\sigma_0=0.2$, $\mu_1=0.1$, $\sigma_1=0.1$,  $a_0=a_1=1$, and $\rho_0=\rho_1=10$. The assigned value of $b=1$ may not be realistic; $b=1$ means that the bank uses $e^{-1}=0.3679$ as its leverage ratio (i.e. (Debt)/(Total Asset)) during the normal business period. This is obviously too restrictive for the banks. Nontheless, we
use this number (and others) since function values are too sensitive to conduct comparative statics analysis with more realistic values like $b=-\log 0.8$.

In Figure 1, when we fix $a$ at $0.3$, $b^*=0.5401$ is the optimal PCA trigger level where $C_1(0,0;b')$ is minimized, and the corresponding cost is $C_1(0,0;b^*)=3.019$. An interesting character about $C_1(0,0;b')$ is that $C_1(0,0;b')$ diverges to infinity on $b'>0.6701$. Note that  at $b'=0.6701$, we have $W_0^{(q)'}(b')/W_0^{(q)}(b')=1$; see Remark \ref{convergence}.  This can be interpreted as follows: With a large $b'$, since it is not likely to initiate a PCA, $X$ shall be at very high level when a PCA is actually applied and therefore the initial cost part and the penalty part become too large. This result suggests the importance of discrete choice of PCA trigger level to prevent the bank from becoming too costly to rescue.  On the other hand, the reason that the cost is relatively high when $b'$ is near to $a$ is that, in these situations, PCA may be applied too early for the cost to be discounted enough.

\begin{figure}[h]
\begin{center}
\begin{minipage}{0.45\textwidth}
\centering{\includegraphics[scale=0.55]{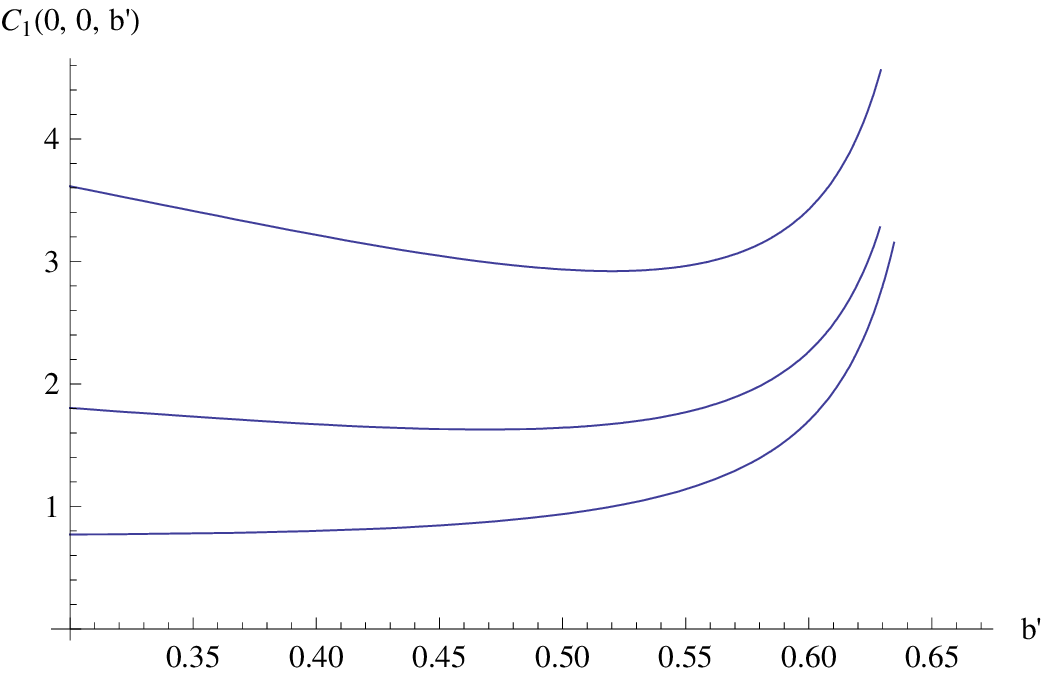}}\\
(i)$\sigma_1=0.1, 0.2, 0.4$ from top to bottom.
\end{minipage}
\begin{minipage}{0.45\textwidth}
\centering{\includegraphics[scale=0.55]{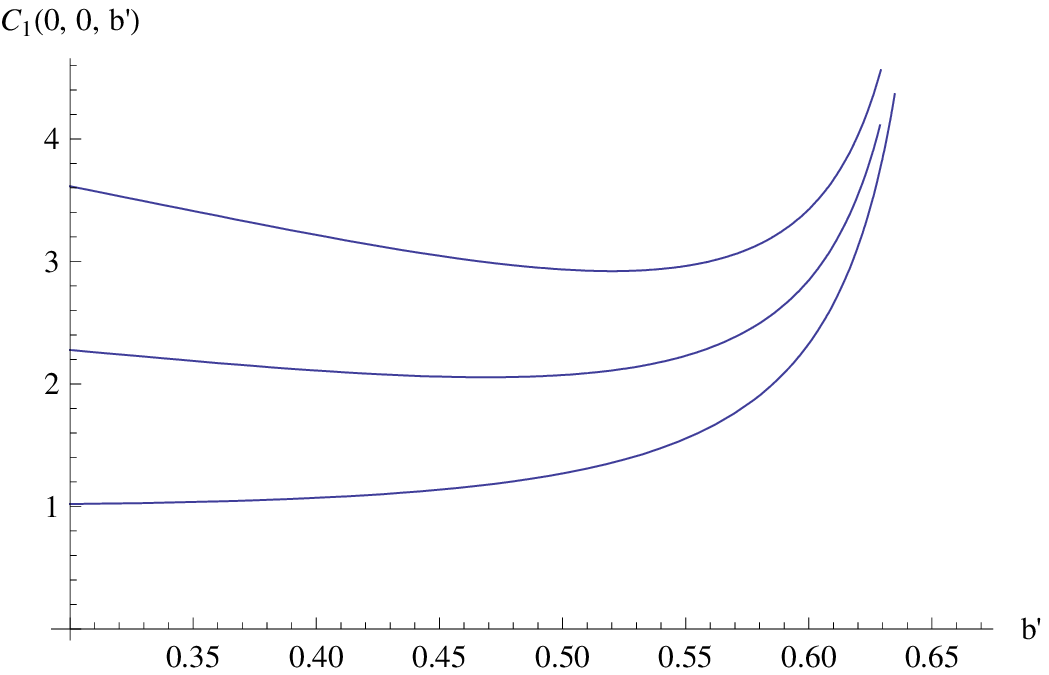}}\\
(ii)$\mu_1=0.1, 0.2, 0.4$ from top to bottom.
\end{minipage}
\caption{}
\end{center}
\end{figure}

Figure 2 (i) and (ii) are the results of some comparative statics: we vary $\sigma_1=0.1, 0.2, 0.4$ and $\mu_1=0.1, 0.2, 0.4$, respectively (the other parameters remain the same).
These graphs show that the increase of $\sigma_1$ and $\mu_1$ raises the optimal threshold $b^*$ and lower the optimized cost $C_1(0,0;b^*)$. The results from  various volatility parameters is worth mentioning.  With a larger $\sigma_1$, the speed to reach either $S$ or $S-b$ shall become greater, so that the running cost shall  be smaller.  The reason for a higher PCA trigger level with a higher $\sigma_1$ is so simple since there are several effects are involved.  However, it may be interpreted this way: since a higher volatility may increase the danger of becoming insolvent and ending up paying penalty (i.e., reaching $s-b$ after PCA was implemented), it would be safer to start a PCA earlier.

\begin{table}[htb]
  \begin{tabular}{|c|c|c|c|c|c|c|} \hline
    $a$ & 0.1 & 0.2 & 0.3 & 0.4 & 0.5 & 0.6 \\ \hline
    $b^*$ & 0.4593 & 0.5136 & 0.5401& 0.5566 & 0.5675 & 0.6 \\ \hline
    $C_1(0,0,b*)$ & 1.992 & 2.690 & 3.019 & 3.120 & 3.064 & 2.973  \\ \hline
    \end{tabular}\\
\caption{Changes of $b^*$ and $C_1(0,0;b^*)$}
\end{table}
\begin{figure}[h]
\begin{center}
\begin{minipage}{1\textwidth}
\centering{\includegraphics[scale=1.1]{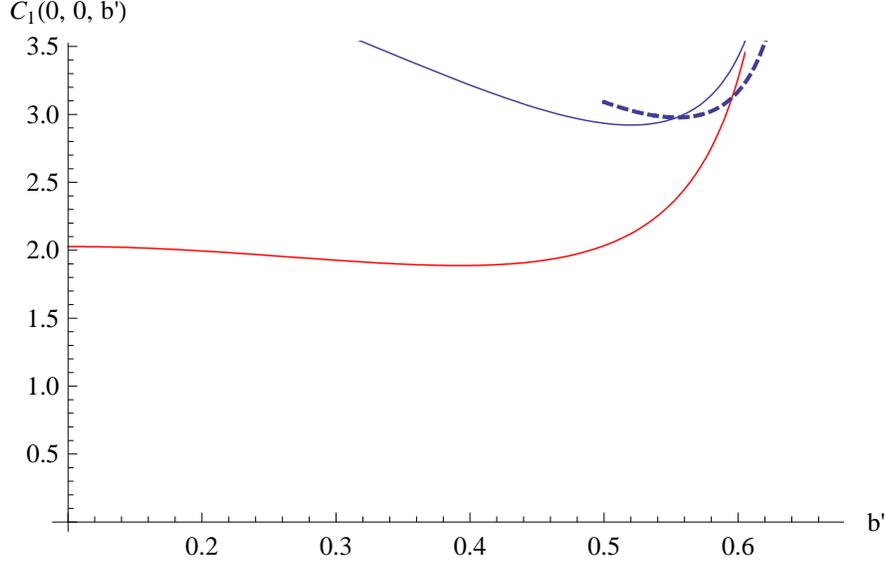}}
\end{minipage}
\caption{the graphs of $C_1(0,0;b')$ with $a=0.1,0.3,0.5$.}
\end{center}
\end{figure}

Finally, we change the level of $a$.  Table 1 shows $b^*$ and $C_1(0,0;b^*)$ with $a=0.1,0.2,\ldots,0.6$, and Figure 3 is the graphs of $C_1(0,0;b')$ with these $a$'s. The optimal threshold $b^*$ moves to the same direction with $a$, but the amount of change is small compared to that of $a$. On the other hand, the cost $C_1(0,0;b^*)$ doesn't change monotonously.  According to the table, around $a=0.4$, the cost has a local maximum.  One of interesting results is that when $a=0.6$, $b^*$ is a boundary solution $0.6$. This means that if $X$ continuously crosses the level $S-b'$, indeed PCA is applied but the regulator doesn't pay initial cost and there doesn't occur positive jump of $X$.

\section{Extension to multiple PCA's}\label{sec:multiple}%%%%%%%%%%%%%%%%%%%%%%%%%%%%%%%%%%%%%%%%%%%%%%%%%%%%%%%%%%%%%%%%%%%%%%%
We considered so far that PCA is applied only once and calculated the cost associated with it. However, after the  bank recovers its leverage ratio to $e^{-b}$ thanks to a PCA, it can be under the regulator's control again when the leverage ratio deteriorates to $e^{b'-b}$ or worse ($S-X\geq b'$). Now we incorporate the possibility that PCA's are  repeatedly applied until the bank becomes finally insolvent. With the method we shall provide here, while it is not of an explicit form,
one can recursively calculate the cost for multiple PCA's.
For a mathematical representation, we redefine the process $I$ by
\[
I(t)=\1_{\{\tau_1\leq t<\tau_2\}}+\1_{\{\tau_3\leq t<\tau_4\}}+ \cdots + \1_{\{\tau_{2n-1}\leq t<\tau_{2n}\}} +\cdots, \quad t\in\R_+
\]
where $\tau_n, \, n=1,2,\ldots$ are $\Fb$-stopping times defined recursively by $\tau_1=T_{b'}$,
\begin{eqnarray*}
\tau_{2n}&=&\tau_{2n-1}+(\tau^+_a\wedge\tau^-_{a-b})\circ \theta_{\tau_{2n-1}},\\
\tau_{2n+1}&=&\tau_{2n} + T_{b'}\circ \theta_{\tau_{2n}},
\end{eqnarray*}
and $(\theta_t)_{t\in\R_+}$ is shift-operator. This definition means that the bank goes under the $n$th PCA at time $\tau_{2n-1}$, and recovers or becomes insolvent at time $\tau_{2n}$. Note that $\tau=\tau_2$.

Additionally, we need the asset values at time $\tau_{2n-1}$ before pushed up, so in the same way as Section 3, we define $\underline{X}_n=X_{\tau_{2n-1}}$\, for $n=1,2,\ldots$ and then, redefine $X_{\tau_{2n-1}}=S_{\tau_{2n-1}}-a$.

Then the cost $C_n(x,s;b')$ for the $n$th PCA can be represented by
\begin{eqnarray*}
C_n(x,s,b')&=&\E^{x,s}\Bigg[\1_{A_n} \bigg(e^{-q\tau_{2n-1}}\left(e^{S_{\tau_{2n-1}}-a}-e^{\underline{X}_n}\right)+\alpha\int_{\tau_{2n-1}}^{\tau_{2n}}e^{-qt}\diff t \\
&&+ e^{-q\tau_{2n}} \left(\beta e^{S_{\tau_{2n-1}}-b}\right)\1_{\{\tau^+_{a}\circ \theta_{\tau_{{2n-1}}} > \tau^-_{a-b}\circ \theta_{\tau_{{2n-1}}}\}}\bigg)\Bigg],
\end{eqnarray*}
where $A_n$ is the event in which the bank is under regulator's strict supervision more than $n$ times until insolvency; that is, $A_n$ can be written by $A_1=\Omega$ and $A_n=\bigcap_{k=1}^{n-1}\{\tau^+_{a}\circ \theta_{\tau_{2k-1}} < \tau^-_{a-b}\circ \theta_{\tau_{2k-1}}\}$ for $n=2,3,\ldots$,
and the total cost $C(x,s,b')$ is given by
\begin{eqnarray*}
C(x,s;b')&=&\sum_{n=1}^{\infty}C_n(x,s;b').
\end{eqnarray*}

\begin{proposition}\label{Cn_1}
If $S_0=X_0=s$, then
\begin{align}\label{eq:Cn_1}
C_{n+1}(s,s;b')&=\frac{W_1^{(q)}(b-a)}{W_1^{(q)}(b)}\int_s^{\infty}\diff m \exp\left(-(m-s)\frac{W_0^{(q)'}(b')}{W_0^{(q)}(b')}\right)C_n(m,m;b')\\
&\times \left(\frac{\sigma^2}{2}\left(\frac{(W_0^{(q)'}(b'))^2}{W_0^{(q)}(b')}-W_0^{(q)''}(b')\right)+\iint_E\Pi(\diff h)\diff y \left(W_0^{(q)'}(y)-\frac{W_0^{(q)'}(b')}{W_0^{(q)}(b')}W_0^{(q)}(y)\right)\right),\nonumber
\end{align}
for $n=1,2,\ldots$ .
\end{proposition}

\begin{proof}
The first-time PCA ends at time $t=\tau_2$. Since $X^1$ only have negative jumps by definition, we have $S_{\tau_2}=S_{T_{b'}}$. Hence
\begin{eqnarray*}
C_{n+1}(s,s;b')&=&
\E^{s,s}\left[ e^{-q\tau_2}\1_{A_2}C_n(S_{\tau_2},S_{\tau_2};b') \right] \\
&=&\E^{s,s}\left[ e^{-qT_{b'}}C_n(S_{T_{b'}},S_{T_{b'}};b') \left( e^{-q(\tau^+_a\wedge\tau^-_{a-b})}\1_{\{\tau^+_a<\tau^-_{a-b}\}} \right)\circ\theta_{T_b'} \right]\\
&=& \E^{s,s}\left[e^{-qT_{b'}}C_n(S_{T_{b'}},S_{T_{b'}};b')\right]\E^{s, s}\left[ \left( e^{-q(\tau^+_a\wedge\tau^-_{a-b})}\1_{\{\tau^+_a<\tau^-_{a-b}\}} \right)\circ\theta_{T_b'}  \Big|\mathcal{F}_{T_b'}\right]\\
&=&\E^{s-a, s}\left[ e^{-q(\tau^+_a\wedge\tau^-_{a-b})}\1_{\{\tau^+_a<\tau^-_{a-b}\}} \right]\Bigg(\int_s^{\infty}\E^{s,s}\left[e^{-qT_{b'}}\1_{\{ S_{T_{b'}}-\underline{X}=b', S_{T_{b'}}\in \diff m \}}\right]C_1(m,m;b')\\
&&+\iiint_D\E^{s,s}\left[e^{-qT_{b'}}\1_{\{ S_{T_{b'}-}-X_{T_{b'}-}\in\diff y, \underline{X}-X_{T_{b'}-}\in \diff h, S_{T_{b'}-}\in \diff m \}}\right]C_1(m,m;b')\Bigg).
\end{eqnarray*}
From Theorem 1 and 2 in Pistrorius \cite{Pistorius_2007}, we have \eqref{eq:Cn_1}.
\end{proof}
As for the finiteness of $C_n(s,s;b')$, the following remark can be shown in the same way as Remark \ref{convergence}.
\begin{remark}
{\rm
If $C_n(m,m;b')<\infty$ for $m\in[s,\infty)$, then $C_{n+1}(s,s;b')<\infty$. Hence by Remark \ref{convergence}, if $ 1-W_0^{(q)'}(b')/W_0^{(q)}(b')<0$, then $C_n(s,s;b')<\infty$ for every $n\geq 1$.
}
\end{remark}

Since we already calculated $C_1$ in the previous subsection, $C_n$ is obtained by repeatedly using this proposition when $S_0=X_0$. The following two propositions are for the other cases; $S_0-X_0\geq b'$ and $S_0-X_0\in(0,b')$. We skip the proofs here since the essential techniques used  are the same as in the propositions above. Note that the results of Proposition \ref{Cn_1} and Proposition \ref{Cn_2} are needed for the calculation in Proposition \ref{Cn_3}.
\begin{proposition}\label{Cn_2}
If $S_0-X_0\geq b'$, then
\begin{align}\label{eq:Cn_2}
C_{n+1}(x,s;b')=\frac{W_1^{(q)}(b-a)}{W_1^{(q)}(b)}C_n(s,s;b').
\end{align}
\end{proposition}

\begin{proposition}\label{Cn_3}
If $S_0-X_0\in(0,b')$, then
\begin{align}\label{eq:Cn_3}
C_n(x,s;b')&= \frac{\sigma^2}{2}\left(W_0^{(q)}(b'-s+x)-\frac{W_0^{(q)'}(b')}{W_0^{(q)}(b')}W_0^{(q)}(b'-s+x)\right)C_n(s-b',s;b')\\
&+\iint_E \Pi(\diff h)\diff y\left( \frac{W_0^{(q)}(b'-s+x)}{W_0^{(q)}(b')}W_0^{(q)}(y)-W_0^{(q)}(y-s+x) \right)C_n(s-y+h,s;b')\nonumber\\
&+\frac{W_0^{(q)}(b-s+x)}{W_0^{(q)}(b)}C_n(s,s;b').\nonumber
\end{align}
where $C_n(\cdot, \cdot; b')$'s on the right-hand side can be computed by Propositions \ref{Cn_1} and \ref{Cn_2}.
\end{proposition}

\section{Appendix} \label{subsec:scale_functions}%%%%%%%%%%%%%%%%%%%%%%%%%%%%%%%%%%%%%
\subsection{Scale functions} Associated with every spectrally negative \lev process, there exists a (q-)scale function
%
%Associated with every spectrally negative \lev process, there exists
%a \emph{(q-)scale function}
\begin{align*}
W^{(q)}: \R \mapsto \R; \quad q\ge 0,
\end{align*}
that is continuous and strictly increasing on $[0,\infty)$ and is
uniquely determined by
\begin{align*}%\label{eq:scale}
\int_0^\infty e^{-\beta x} W^{(q)}(x) \diff x = \frac 1
{\psi(\beta)-q}, \qquad \beta > \Phi(q).
\end{align*}
%where
%\begin{align*}
%\Phi(q) = \sup \left\{ \lambda > 0: \psi(\lambda) = q \right\}, \quad q \geq 0.
%\end{align*}

Fix $a > x > 0$.  If $\tau_a^+$ is the first time the process goes above $a$ and $\tau_0$ is the first time it goes below zero, then we have
\begin{align}\label{eq:exit-time}
\E^x \left[ e^{-q \tau_a^+} 1_{\left\{ \tau_a^+ < \tau_0, \, \tau_a^+ < \infty
\right\}}\right] = \frac {W^{(q)}(x)}  {W^{(q)}(a)} \quad
\textrm{and}  \quad \E^x \left[ e^{-q  \tau_0} 1_{\left\{ \tau_a^+ >
 \tau_0, \, \tau_0 < \infty \right\}}\right] = Z^{(q)}(x) - Z^{(q)}(a) \frac {W^{(q)}(x)}
{W^{(q)}(a)},
\end{align}
where
\begin{align*}%\label{eq:Z-q-function}
Z^{(q)}(x) := 1 + q \int_0^x W^{(q)}(y) \diff y, \quad x \in \R.
\end{align*}
Here we have
\begin{equation}\label{eq:at-zero}
W^{(q)}(x)=0 \quad\text{on} \quad  (-\infty,0)\quad\text{ and}\quad
Z^{(q)}(x)=1 \quad\text{on}\quad (-\infty,0].
\end{equation}
We also have
\begin{align}
\E^x \left[ e^{-q  \tau_0} \right] = Z^{(q)}(x) - \frac q {\Phi(q)} W^{(q)}(x), \quad x > 0. \label{laplace_tau_0}
\end{align}

 In particular, $W^{(q)}$ is continuously differentiable on $(0,\infty)$ if $\Pi$ does not have atoms and $W^{(q)}$ is twice-differentiable on $(0,\infty)$ if $\sigma > 0$; see, e.g., \cite{Chan_2009}.  Throughout this paper, we assume the former.
\begin{assump}
We assume that $\Pi$ does not have atoms.
\end{assump}

Fix $q > 0$.  The scale function increases exponentially;
\begin{align}
W^{(q)} (x) \sim \frac {e^{\Phi(q) x}} {\psi'(\Phi(q))} \quad
\textrm{as } \; x \uparrow \infty.
\label{scale_function_asymptotic}
\end{align}
There exists a (scaled) version of the scale function $ W_{\Phi(q)}
= \{ W_{\Phi(q)} (x); x \in \R \}$ that satisfies
\begin{align}
W_{\Phi(q)} (x) = e^{-\Phi(q) x} W^{(q)} (x), \quad x \in \R \label{W_scaled}
\end{align}
and
\begin{align*}
\int_0^\infty e^{-\beta x} W_{\Phi(q)} (x) \diff x &= \frac 1
{\psi(\beta+\Phi(q))-q}, \quad \beta > 0.
\end{align*}
Moreover $W_{\Phi(q)} (x)$ is increasing, and as is clear from
\eqref{scale_function_asymptotic},
\begin{align}
W_{\Phi(q)} (x) \uparrow \frac 1 {\psi'(\Phi(q))} \quad \textrm{as }
\; x \uparrow \infty. \label{scale_function_asymptotic_version}
\end{align}

Regarding its behavior in the neighborhood of zero, it is known that
\begin{align}
W^{(q)} (0) = \left\{ \begin{array}{ll} 0, & \textrm{unbounded
variation} \\ \frac 1 {\mu}, & \textrm{bounded variation}
\end{array} \right\} \quad \textrm{and} \quad W^{(q)'} (0+) =
\left\{ \begin{array}{ll}  \frac 2 {\sigma^2}, & \sigma > 0 \\
\infty, & \sigma = 0 \; \textrm{and} \; \Pi(0,\infty) = \infty \\
\frac {q + \Pi(0,\infty)} {\mu^2}, & \textrm{compound Poisson}
\end{array} \right\}; \label{at_zero}
\end{align}
see Lemmas 4.3-4.4 of
\cite{Kyprianou_Surya_2007}.
For a comprehensive account of the scale function, see
\cite{Bertoin_1996,Bertoin_1997, Kyprianou_2006, Kyprianou_Surya_2007}. See \cite{Egami_Yamazaki_2010_2, Surya_2008} for numerical methods for computing the
scale function.

\bibliographystyle{plain}
\def\cprime{$'$}

\end{document}